\documentclass[11pt]{article}
\pdfoutput=1
\usepackage{amsmath,amssymb,latexsym,epsfig,amsthm}
\usepackage{times}
\usepackage{color}
\usepackage{adjustbox}
\usepackage{enumerate}
\usepackage{appendix}
\usepackage[bottom]{footmisc}
\usepackage{algorithm}
\usepackage[noend]{algorithmic}
\usepackage[utf8]{inputenc}
\usepackage{authblk}

\usepackage{subcaption}
\usepackage{pgf,tikz}
\usepackage[pdfpagelabels,plainpages=false,hypertexnames=false]{hyperref}%
\hypersetup{%
   breaklinks,%
   ocgcolorlinks,
   colorlinks=true,%
   linkcolor=[rgb]{0.45,0.0,0.0},%
   urlcolor=[rgb]{0.05,0.390,0.0},%
   citecolor=[rgb]{0,0,0.45}
}%

\usetikzlibrary{arrows,shapes,calc,intersections,through,backgrounds}
\tikzset{faded/.style={color = black, densely dotted}}
\tikzset{help/.style={color = black, thin, loosely dotted}}
\tikzset{OPT/.style={color = red, thick, dashed}}
\definecolor{xdxdff}{rgb}{0.49,0.49,1}
\definecolor{ffqqqq}{rgb}{1,0,0}
\definecolor{qqqqff}{rgb}{0,0,1}
\definecolor{zzttqq}{rgb}{0.6,0.2,0}

\newtheorem{theorem}{Theorem}
\newtheorem{lemma}{Lemma}

\def\Real{\mathbb R}

\def\C{{\cal C}}
\def\Cint{{\cal C}^{\text{int}}}
\def\Cext{{\cal C}^{\text{ext}}}
\def\Cbad{{\cal C}^{\text{bad}}}
\def\Cnew{{\cal C}^{\text{new}}}
\def\D{{\cal D}}
\def\polygon{P}
\def\barsigma{\overline{\sigma}}
\def\barSigma{\overline{\Sigma}}
\def\n{n}
\def\conv{\text{Conv}}
\def\pb{\bar{p}}
\def\p{p}
\def\Sb{\bar{S}}

\def\Dnew{{\cal D}_{\text{new}}}
\def\Dbad{{\cal D}_{\text{bad}}}
\def\F{{\cal F}}
\def\B{{\cal B}}
\def\compcover{\text{compcover}}
\def\decompose{\text{decompose}}
\def\eps{\varepsilon}

\begin{document}

\begin{titlepage}
\title{Approximation Schemes for Partitioning: Convex Decomposition and 
Surface Approximation}

\author{Sayan Bandyapadhyay\thanks{sayan-bandyapadhyay@uiowa.edu} \qquad Santanu Bhowmick\thanks{santanu-bhowmick@uiowa.edu} \qquad Kasturi Varadarajan\thanks{kasturi-varadarajan@uiowa.edu}}
\affil{Department of Computer Science,\\ University of Iowa,\\ Iowa City, USA}
\maketitle

\begin{abstract}
We revisit two NP-hard geometric partitioning problems -- convex decomposition and surface approximation. Building on recent developments in geometric 
separators, we present quasi-polynomial time algorithms for these problems
with improved approximation guarantees. 
\end{abstract} 
\end{titlepage}

\section{Introduction}\label{sec:intro}
The size and complexity of geometric objects are steadily growing due to the technological advancement of the tools that generate these. A simple strategy to deal with large, complex objects is to model them using pieces which are easy to handle. However, one must be careful while applying this kind of strategy as such decompositions are costly to construct and may generate a plethora of components. In this paper, we study two geometric optimization problems which deal with representations of complex models using simpler structures.
  
\subsection{Polygon Decomposition Problem}
The problem of decomposing a polygon into simpler pieces has a wide range of applications in VLSI, Robotics, Graphics, and Image Processing. Decomposition of complex polygons into convex pieces make them suitable for applications like skeleton extraction, mesh generation, and many others \cite{motion_planning2,motion_planning1,skeletonization}. Considering the importance of this problem it has been studied for more than thirty years. Different versions of this problem have been considered based on the way of decomposition. The version we consider in this paper, which we refer to as the \textit{convex decomposition problem}, is defined as follows.

\paragraph{Convex Decomposition Problem:} Let $P$ be a polygon possibly with polygonal holes. A \textit{diagonal} in such a polygon is a line segment connecting two non-consecutive vertices that lies inside the polygon. The \textit{Convex Decomposition Problem} is to add a set of non-crossing diagonals so that each subpolygon in the resulting decomposition of $P$ is convex, and the number of subpolygons is minimized.

For the polygons without holes, this problem can be solved in $O(r^2n\log n)$ time using a dynamic programming based approach, where $r$ is the number of reflex vertices \cite{Keil85}. Additionally, the running time could be improved to $O(n+ r^2\min\{r^2,n\})$ \cite{keils02}. However, the problem is $\mathcal{NP}$-Hard if the polygon contains holes \cite{nphard82}. Hence Chazelle \cite{cd_holes} gave a 4.333 factor approximation algorithm by applying a separator theorem recursively. Later, Hertel and Mehlhorn \cite{cd_4factor} improved the approximation factor to 4 by applying a simple strategy based on triangulation.

Many other versions of the polygon decomposition problem have also been studied. One of these versions allows the algorithm to add additional points inside the polygon. The endpoints of the line segments which decompose the polygon can be chosen from these additional points and the polygon vertices. Considering this version Chazelle and Dobkin~\cite{simple_cd} have designed an $O(n^3)$ time optimal algorithm for the simple polygons without holes. Later, they improved the running time to $O(n+r^3)$, where $r$ is the number of reflex angles \cite{chazzele_dobkin2,Chazzele_thesis}. But, the problem is still $\mathcal{NP}$-Hard for polygons with holes \cite{nphard82}. In a different version the polygon is decomposed into \textit{approximately} convex pieces where concavities are allowed within some specified tolerance \cite{ACD1,ACD4,ACD2,ACD3}. The idea is that these approximate convex pieces can be computed efficiently and can result in a smaller number of pieces. Another interesting variant is to consider an additional set of points inside the polygon and try to find the minimum number of convex pieces such that each piece contains at most $k$ such points, for a given $k$ \cite{irina_thesis}.

\subsection{Surface Approximation Problem}
In many scientific disciplines including Computer Graphics, Image Processing, and Geographical Information System (GIS), surfaces are used for representation of geometric objects. Thus modeling of surfaces is a core problem in these areas, and polygonal descriptions are generally used for this purpose. However, considering the complexity of the input objects the goal is to use a minimal amount of polygonal description. 

In many scientific computations 3D-object models are used. In that case the surface is approximated using piecewise linear patches (i.e, polygonal objects) whose vertices are allowed to lie within a close vicinity of the actual surface. To ensure that the local features of the original surface are retained, one may end up generating an unmanagable number of patches. But, this is not at all cost effective for real time applications. Thus one can clearly note the complexity-quality tradeoff in this context.

We consider the surface approximation problem for $xy$-monotone surfaces. The original surface is the graph of a continuous bivariate function $f(x,y)$ whose domain is $\Real^2$. The goal is to compute a piecewise-linear function $g(x,y)$ which approximate the function $f(x,y)$. The domain of $g(\cdot,\cdot)$ is also
$\Real^2$, and we wish to minimize the number of its faces, which we require
to be triangles. We formally define the problem as follows.

\paragraph{Surface Approximation Problem:} Let $f$ be a bivariate function and $\Sb$ be a set of $n$ points sampled from $f$. Given $\mu > 0$ a piecewise linear function $g$ is called an $\mu$-\textit{approximation} of $f$ if $$|g(x_i,y_i)-z_i| \leq \mu$$ for every point $(x_i,y_i,z_i) \in \Sb$. The surface approximation problem is to compute, given $f$ and $\mu$, such a $g$ with minimum complexity, where the complexity of a piecewise linear surface is defined to
be the number of its faces.

Considering its importance there has been a lot of work on this problem in computer graphics and image processing \cite{heur1,heur2,heur3,heur4}. However, most of these approaches are based on heuristics and don't give any guarantees on the solution. There are two basic techniques which are used by these algorithms: \textit{refinement} and \textit{decimation}. The former method starts with a triangle and further refine it locally until the solution becomes an $\mu$-\textit{approximation}. The latter starts with a triangulation and coarsens it locally until one can't remove a vertex \cite{heur5,heur6,heur7}. 

The first provable guarantees for this problem were given by Agarwal and Suri. They gave an algorithm which computes an approximation of size $O(c\log c)$, and runs in $O(n^8)$ time, where $c$ is the complexity of the optimal $\mu$-\textit{approximation} \cite{surface_approx_suri}. They also proved that the decision version of this problem is $\mathcal{NP}$-complete. In a later work Agarwal and Desikan \cite{surface_approx_desikan} presented a randomized algorithm which computes an approximation of size $O(c^2\log^2 c)$ in $O(n^{2+\delta}+c^3\log^2 c\log \frac{n}{c})$ expected time, where $c$ is the complexity of the optimal $\mu$-\textit{approximation} and $\delta$ is any arbitrary small positive number. 

A different version of the surface approximation problem has also considered by Mitchell and Suri \cite{MitchellSuri}. Given a convex polytope $P$ with $n$ vertices and $\mu > 0$ they designed an $O(n^3)$ time algorithm which computes a $O(c\log n)$ size convex polytope $Q$ such that $(1-\mu)P \leq Q \leq (1+\mu)P$, where $c$ is the size of such an optimal polytope. Later, this bound was improved independently by Clarkson \cite{Clarkson} and Bronniman and Goodrich \cite{Bronnimann}. The latter presented an $O(nc(c+\log n)\log \frac{n}{c})$ time algorithm which computes a polytope $Q$ of size $O(c)$. The hardness bound for this version is still an open problem. 

\subsection{Our Results}
We obtain a quasi-polynomial time approximation scheme (QPTAS) for the convex decomposition problem using a separator based approach. We show the existence of a suitable set of diagonals (our separator), which partitions the optimal solution in a balanced manner, and intersects with a small fraction of optimal solution. Moreover, the set of diagonals can be guessed from
a quasi-polynomial sized family. We then show the existence of a near-optimal
solution that respects the set of diagonals. As we explain below, this proof of
the existence of a suitable diagonal set, and a near-optimal solution that respects it, are our main technical contribution.  

The approximation scheme is now a straightforward application: guess a separator, recurse on the subpolygons. The recursion bottoms out when we reach a subpolygon for which the optimal convex decomposition has a small size. This base case can be detected and solved in quasi-polynomial time by an exhaustive search. 

Our result builds on the recent breakthrough due to Adamaszek and Wiese \cite{AdamaszekW13,AdamaszekW14}. They presented a QPTAS for independent set of weighted axis parallel rectangles \cite{AdamaszekW13}, and subsequently extended their approach to polygons with polylogarithmic many vertices \cite{AdamaszekW14}. 
Har-peled \cite{Har-Peled13} simplified and generalized their result to polygons of arbitrary complexity. Mustafa \textit{et al.} \cite{MRS14} also describe a simplification, other generalizations, and an application to obtain a QPTAS for computing a minimum weight set cover using pseudo-disks.

For our problem, these results \cite{AdamaszekW14,Har-Peled13,MRS14} imply a separator that
intersects a small number of convex polygons in the optimal decomposition and
partitions the remaining convex polygons evenly. However, such a separator
may pass through the holes of polygon, and its intersection with the polygon
may not be a set of diagonals. How do we convert the separator into a set
of diagonals that still partitions nicely? And how do we convert the optimal
decomposition into a near-optimal decomposition that respects this diagonal set?
These are the key questions we address in our work.
 
For the surface approximation problem, we describe a 
quasi-polynomial time algorithm that computes a 
surface whose complexity is within a multiplicative constant factor 
of the optimal surface. The main contribution is a reduction from the surface
approximation problem to a planar problem of computing a disjoint set cover using a certain family of triangles. We design a QPTAS for the disjoint set cover problem using the separator based approach. While similar reductions have been used in previous work on this problem, our reduction increases the size of the solution by at most a (multipilicative) constant factor. Our family of triangles 
has a useful closure property that facilitates the working of the QPTAS. In
this way, we get a quasi-polynomial time $O(1)$-factor approximation algorithm for the surface approximation problem.

\section{Convex Decomposition}

Recall that we are given a polygon $\polygon$ with zero or
more polygonal holes. A diagonal in such a polygon is a line segment between
two non-adjacent vertices that lies entirely within the polygon. The problem is to add a set of non-crossing
diagonals so that each subpolygon in the resulting decomposition of $\polygon$
is convex, and the number of convex polygons is the minimum possible.

\subsection{A Separator}
A set $D$ of diagonals of polygon $\polygon$ is said to be conforming for
$\polygon$ if no two diagonals in $D$ cross. The diagonals in $D$ naturally
partition $\polygon$ into a set of polygons $\{P_1, P_2, \ldots, P_t\}$.  

Let $K(\polygon)$ denote the number of convex polygons in an optimal 
diagonal-based convex decomposition of $\polygon$.

\begin{lemma}
\label{lem:sep1}
Let $\polygon$ be a polygon (possibly with holes) that has $n$ vertices, let
$K = K(\polygon)$, and $0 < \delta < 1$ be a parameter. If $K \geq \frac{ c \log (1/\delta)}{\delta^3}$,
where $c$ is a sufficiently large constant, 
we can compute, in $n^{O(1/\delta^2)}$ time, a family 
$\D = \{ D_1, D_2, \ldots, D_t \}$ of conforming sets of diagonals of 
$\polygon$ with the following property: there exists a $1 \leq i \leq t$
such that the diagonals in $D_i$ partition $\polygon$ into polygons 
$P_1, P_2, \ldots, P_s$, so that (a) $K(P_j) \leq (2/3 + \delta) K$, and (b) $\sum_{j=1}^s K(P_j) \leq (1 + \delta)  K$.
\end{lemma}  

The rest of the section is devoted to the proof of the lemma, which proceeds
in three major steps. Fix an optimal diagonal-based convex decomposition of
$\polygon$, and let $\C = \{ C_1, \ldots, C_K \}$ be the set of resulting
convex polygons. In the first step, we argue that there is a polygonal
cycle $\Sigma$ such that (a) the number of polygons in $\C$ that are 
contained inside (resp.\ outside) $\Sigma$ is at most $2K/3$, and (b) the 
number of polygons in $\C$ whose interiors are intersected by $\Sigma$ is
at most $\delta K/30$. This step is similar to the constructions in 
\cite{AdamaszekW14,Har-Peled13,MRS14}, but we need to review some particulars of these constructions to
identify information that will be needed in subsequent steps. Notice that
$\Sigma$ may actually intersect the holes of the polygon $\polygon$.

In the second step, we show how to convert $\Sigma$ into a conforming set
of diagonals $D$ such that as far as being a separator for $\C$ is concerned,
$D$ behaves like a proxy for $\Sigma$. As a result of this process, the 
number of diagonals in $D$ may be significantly larger than the number of
edges in $\Sigma$. Nevertheless, $D$ is an easily computed function of
$\Sigma$.

In the third step, we exploit the separator properties of $D$ to compute,
from $\C$, a suitable convex decomposition of $\polygon$ that respects the 
diagonals in $D$. With this overview, we are ready to describe the three
steps.

\paragraph{Step 1:} We have already fixed an optimal diagonal-based convex decomposition of
$\polygon$, with $\C = \{ C_1, \ldots, C_K \}$ being the set of resulting
convex polygons. Without loss of generality, assume that no two vertices
of $\polygon$ lie on a vertical line. For each convex polygon $C_i$, let
$s_i$ denote the line segment connecting the leftmost and rightmost points
of $C_i$. We call $s_i$ the {\em representative} segment of $C_i$. Notice that
$s_i$ is either an edge of $\polygon$ or a diagonal of $\polygon$. Let
$S = \{s_1, s_2, \ldots, s_K\}$. See parts (a) and (b) of Figure~\ref{fig:step1}.

\begin{figure}[H]
\tikzstyle{trapezoid} = [dotted, semithick]
\tikzstyle{diags} = [dashed]
\tikzstyle{sample} = [blue, thick]
\tikzstyle{separator} = [red,dashdotted,thick]
\centering
  \begin{subfigure}{0.45\textwidth}
  \centering
  \resizebox{\linewidth}{!} {   
   \begin{tikzpicture}
       \coordinate (a) at ( 3, 6);
       \coordinate (b) at ( 2, 1);
       \coordinate (c) at (11, 1);
       \coordinate (d) at ( 9, 7);

       \coordinate (e) at ( 5 , 5);
       \coordinate (f) at (5.5, 4);
       \coordinate (g) at (4.5, 4);
       \coordinate (h) at (4.11,3);
       \coordinate (i) at (5.56,3);
       \coordinate (j) at (4.6, 2);
       \coordinate (k) at (7.5, 2);
       \coordinate (l) at (6.5, 3);
       \coordinate (m) at ( 8, 3);
       \coordinate (n) at (7.38, 4);
       \coordinate (o) at (6.5, 4);
       \coordinate (p) at (7  , 5);

       \coordinate (q) at (1,  0);
       \coordinate (t) at (1,  8);
       \coordinate (r) at (12, 8);
       \coordinate (s) at (12, 0);
       
       \draw (a) -- (b) -- (c) -- (d) -- (a);
       \draw (e) -- (f) -- (g) -- (h) -- (i) -- (j) -- (k) -- (l) -- (m) -- (n) -- (o) -- (p) -- (e);
       \draw[dashed] (a) -- (e); 
       \draw[dashed] (a) -- (g);
       \draw[dashed] (b) -- (h);
       \draw[dashed] (b) -- (j);
       \draw[dashed] (c) -- (k);
       \draw[dashed] (c) -- (m);
       \draw[dashed] (d) -- (p);
       \draw[dashed] (d) -- (n);
       \draw[opacity = 0] (q) -- (t) -- (r) -- (s) -- (q);
   \end{tikzpicture}
   }
   \caption {Polygon $P$ (with hole) along with its optimal convex decomposition
      $\mathcal{C}$, the latter denoted using dashed lines.
   }
  \end{subfigure}
  \hfill
  \begin{subfigure}{0.45\textwidth}
  \centering
  \resizebox{\linewidth}{!} {   
  \begin{tikzpicture}
      \draw[sample] (a) -- (d);
      \draw[sample] (b) -- (c);
      \draw[sample] (b) -- (i);
      \draw[sample] (d) -- (o);
      \draw[diags] (a) -- (f);
      \draw[diags] (b) -- (g);
      \draw[diags] (c) -- (l);
      \draw[diags] (c) -- (n);

      \draw[opacity = 0] (q) -- (r) -- (s) -- (t) -- (q);
  \end{tikzpicture}
  }
  \caption {The representative segment $s_i$ of each polygon $C_i \in
     \mathcal{C}$ is shown. Segments shown using solid lines are the ones 
     chosen in the sample $R$.
  }
  \end{subfigure}
  \vspace{1cm}
  \tikzstyle{vertex} = [circle, fill=black, minimum size=4pt, inner sep=0pt]
  \begin{subfigure}{0.45\textwidth}
  \centering
  \resizebox{\linewidth}{!} {   
  \begin{tikzpicture}
      \coordinate (u_1) at (2,   8);
      \coordinate (t_1) at (2,   0);
      \coordinate (z_1) at (5.5, -7);
      \coordinate (s_1) at (11, 0);
      \coordinate (v_1) at (3, 8);
      \coordinate (r_1) at (3, 1.56);
      \coordinate (u)   at (5.56, 6.5);
      \coordinate (v)   at (6.5,  6.6);
      \coordinate (z)   at (5.56, 1);
      \coordinate (q_1) at (6.5, 1);
      \coordinate (w_1) at (9, 8);
      \coordinate (w) at (9, 1);
      \coordinate (s_1) at (11, 8);
      \coordinate (z_1) at (11, 0);

      \draw[sample] (a) -- (d);
      \draw[sample] (b) -- (c);
      \draw[sample] (b) -- (i);
      \draw[sample] (d) -- (o);

       \draw (q) -- (t) -- (r) -- (s) -- (q);

      \draw[trapezoid] (u_1) -- (t_1);
      \draw[trapezoid] (i) -- (u);
      \draw[trapezoid] (o) -- (v);
      \draw[trapezoid] (o) -- (q_1);
      \draw[trapezoid] (c) -- (s_1);
      \draw[trapezoid] (d) -- (w_1);
      \draw[trapezoid] (d) -- (w);
      \draw[trapezoid] (c) -- (z_1);

      \draw[separator] (i) -- (z);
      \draw[separator] (o) -- (q_1);
      \draw[separator] (o) -- (d);
      \draw[separator] (d) -- (w_1);
      \draw[separator] (v_1) -- (r_1);
      \draw[separator] (v_1) -- (w_1);
      \draw[separator] (r_1) -- (i);
      \draw[separator] (q_1) -- (z);

      \node[vertex, label ={-80:{a}}] at (v_1) {$ $}; 
      \node[vertex, label ={-80:{b}}] at (r_1) {$ $}; 
      \node[vertex, label ={-80:{c}}] at (i) {$ $}; 
      \node[vertex, label ={-80:{d}}] at (z) {$ $}; 
      \node[vertex, label ={-80:{e}}] at (q_1) {$ $}; 
      \node[vertex, label ={-80:{f}}] at (o) {$ $}; 
      \node[vertex, label ={-80:{g}}] at (d) {$ $};
      \node[vertex, label ={-80:{h}}] at (w_1) {$ $};

  \end{tikzpicture}
  }
  \caption {Trapezoidal decomposition, shown in dotted lines,  of the segments in $R$. The separator $\Sigma$ is the dashed polygonal cycle $a$-$b$-$c$-$d$-$e$-$f$-$g$-$h$.
  }
  \end{subfigure}
  \hfill
  \begin{subfigure}{0.45\textwidth}
  \centering
  \resizebox{\linewidth}{!} {   
  \begin{tikzpicture}
       \draw (a) -- (b) -- (c) -- (d) -- (a);
       \draw (e) -- (f) -- (g) -- (h) -- (i) -- (j) -- (k) -- (l) -- (m) -- (n) -- (o) -- (p) -- (e);
       \draw[separator] (i) -- (z);
       \draw[separator] (o) -- (q_1);
       \draw[separator] (o) -- (d);
       \draw[separator] (d) -- (w_1);
       \draw[separator] (v_1) -- (r_1);
       \draw[separator] (v_1) -- (w_1);
       \draw[separator] (r_1) -- (i);
       \draw[separator] (q_1) -- (z);

       \draw[opacity = 0] (q) -- (t) -- (r) -- (s) -- (q);
  \end{tikzpicture}
  }
  \caption {The seperator $\Sigma$ and the polygon $\polygon$.}
  \end{subfigure}
\caption{Step 1 of the proof of Lemma \ref{lem:sep1}}
\label{fig:step1}
\end{figure}
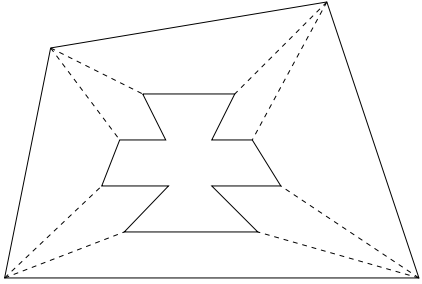
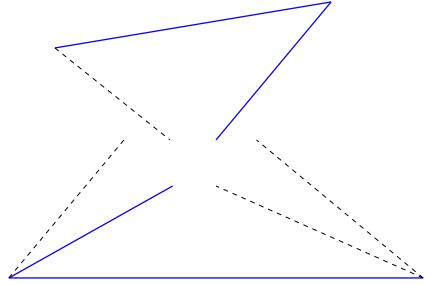
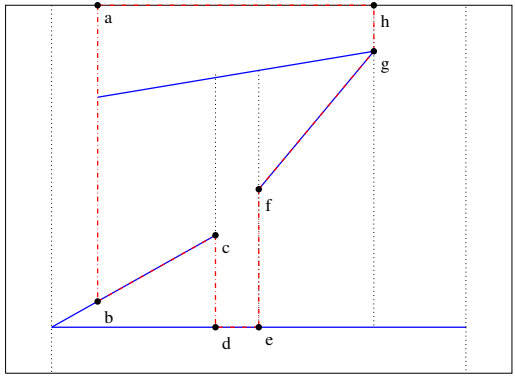
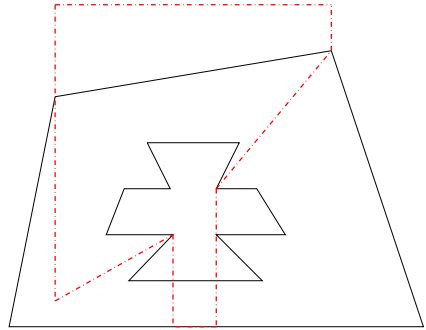

Let us fix an axes parallel box $B$ that contains $\polygon$. Let us pick
a random subset $R \subseteq S$ of size $\lceil r \log r \rceil \leq K$, 
where $r$ is a parameter that we fix below. We compute the 
{\em trapezoidal decomposition}, restricted to $B$,
of the segments in $R$. (See Figure~\ref{fig:step1}(c).) That is, from each endpoint $p$ of a
segment in $R$, we shoot a vertical ray upwards (resp. downwards) till it hits
either one of the other segments in $R$ or the boundary of $B$. We refer to the point
that the ray hits as $u_p$ (resp. $d_p$). The vertical segments
thus generated partition $B$ into faces, each of which is a trapezoid. The 
trapezoidal decomposition can be viewed as a planar graph whose faces correspond to the trapezoids. The vertices of
this planar graph are the endpoints of segments in $R$, the vertices of $B$, and the
points of the form $u_p$ or $d_p$. There are two types of edges -- {\em non-vertical}
and {\em vertical}. An edge that is contained within a segment $s \in R$ is a line 
segment connecting two consecutive vertices that lie in $s$. Similarly, there is an 
edge connecting every two consecutive vertices on the top (and bottom) edge of $B$. 
These edges constitute the non-vertical type. 
The vertical edges include ones of the form $\overline{pu_p}$ (resp. $\overline{pd_p}$), where $p$ is an endpoint
of a segment in $S$. The left and right edges
of $B$, which will also be edges of the planar graph, are included in the vertical category as well.
The number of vertices, edges, and faces of the trapezoidal decomposition is $O(r \log r)$.

Notice that a nonvertical edge lying on the boundary of $B$ does not intersect any convex polygon
in $\C$. Any other nonvertical edge intersects the interior of only the convex polygon
whose representative segment it lies on. No nonvertical edge intersects the interior of
any hole of $\polygon$. A vertical edge, on the other hand, can intersect the interior of several convex
polygons in $\C$ as well as the interiors of several holes. Note however, that a vertical edge
intersects the interior of a convex polygon in $\C$ if and only if it intersects the relative interior of the
corresponding representative segment. By standard sampling theory, there
exists a choice of $R$, such that the number of convex polygons in $\C$ (representative segments
in $S$) intersected by any vertical edge in the trapezoidal decomposition of $R$ is at most $c_1K/r$, for some constant $c_1 > 0$. 
We will assume henceforth that $R$ satisfies this property.

Let $\n(e)$ denote the number of convex polygons whose interiors are intersected by edge $e$ of
the trapezoidal decomposition. 

Recall that the trapezoidal decomposition of $R$ is a planar graph with
$O(r \log r)$ vertices, edges, and faces. The articles \cite{AdamaszekW14,Har-Peled13,MRS14} study separators
which are simple polygonal cycles whose edges are the edges of the planar graph.
In particular, the arguments of \cite{AdamaszekW14,Har-Peled13,MRS14} imply that there is a simple polygonal cycle
$\Sigma$ in the plane, constituted of  $O(\sqrt{r \log r})$ edges of the trapezoidal decomposition,
such that (a) the number of representative segments (convex polygons) in the interior of
$\Sigma$ is at most $2K/3$, and (b) the number of representative segments (convex polygons) in
the exterior of $\Sigma$ is at most $2K/3$. See Figure~\ref{fig:step1}(d). Let $\Cint$ (resp. $\Cext$) denote the subset consisting
of those polygons of $\C$ in the interior (resp. exterior) of $\Sigma$. Abusing notation, 
we say $e \in \Sigma$ to mean that $e$ is an edge of the trapezoidal decomposition that is
contained in $\Sigma$. We have that $\sum_{e \in \Sigma} n(e) \leq O(\sqrt{r \log r}) \frac{c_1K}{r}$.
We will choose $r$ large enough so that 

\begin{equation}
\label{eq:sigma}
30 \sum_{e \in \Sigma} n(e) \leq \delta K.
\end{equation} 

This can be ensured by, say, setting $r = c/\delta^3$ for sufficiently large constant $c$. Then $\Sigma$
would be constituted of $O(1/\delta^2)$ edges of the trapezoidal decomposition. Each vertex of 
such an edge is specified by a tuple consisting of $O(1)$ features of the input polygon -- note
that a vertex of the form of $u_p$ or $d_p$ is specified by $p$ and a representative segment,
which is a diagonal or edge of the input polygon. Thus $\Sigma$ can be specified by $O(1/\delta^2)$ such
tuples. This implies that there is an algorithm, that, given $\polygon$, computes in $n^{O(1/\delta^2)}$ time a 
family of $n^{O(1/\delta^2)}$ cycles that contains a $\Sigma$ satisfying (\ref{eq:sigma}).        
 
\paragraph{Step 2:} We delete from $\Sigma$ the portions that lie in the interiors of the holes (this includes
the unbounded hole outside $\polygon$ as well). That is, we consider $\Sigma \cap \polygon$. We
further partition each connected component of $\Sigma \cap \polygon$ 
using the vertices of $\polygon$ that lie in the relative interior of the
component. See Figure~\ref{fig:step2-1}.
This partitions $\Sigma \cap \polygon$ into {\em fragments}, which are polygonal chains. An endpoint of such
a chain is either a vertex of $\polygon$ or a point $q$ that lies in the interior of an edge $f$ of the
polygon. Let $\Sigma'$ denote the resulting collection of fragments. Each fragment $\sigma$ contains at most two
vertical edges (which would be portions of vertical edges in $\Sigma$) and at most one nonvertical
edge (which would be part of some representative segment, and made up of one or more edges of the trapezoidal decomposition that are contiguous on that segment). For a convex polygon $C \in \C$, 
let $n(\sigma,C)$ denote the number of connected components of $\sigma \cap (\text{interior}\  C)$. The quantity
$n(\sigma,C)$ is either $0$, $1$, or $2$ -- we can get two components if $\sigma$ actually has
two vertical edges that both intersect $C$. Let $n(\sigma) = \sum_{C \in \C} n(\sigma,C)$.

\begin{figure}[hbt]
\centering
\begin{tabular}{c@{\hspace{0.1\linewidth}}c}
\includegraphics[scale=0.3]{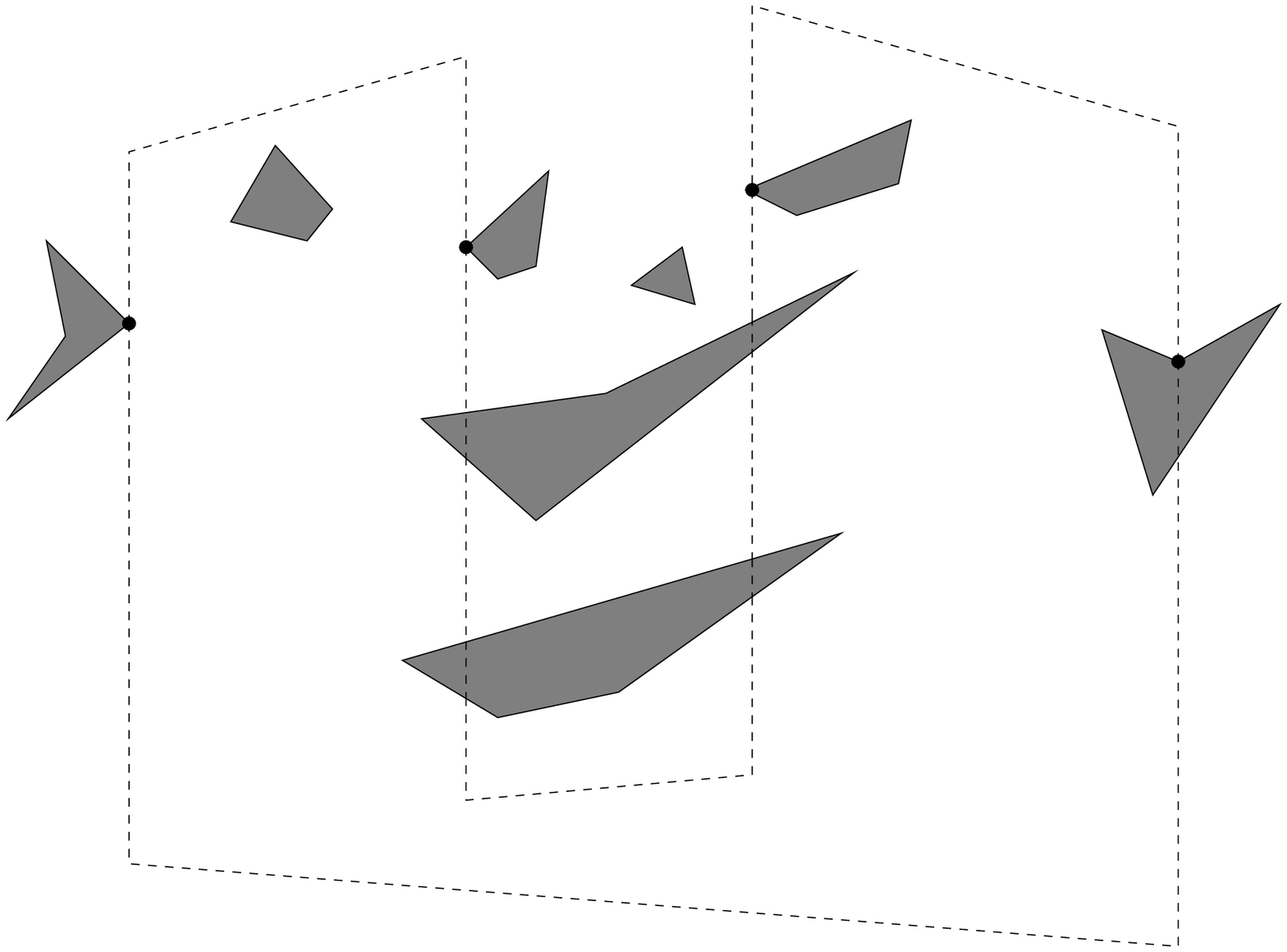} &
\includegraphics[scale=0.3]{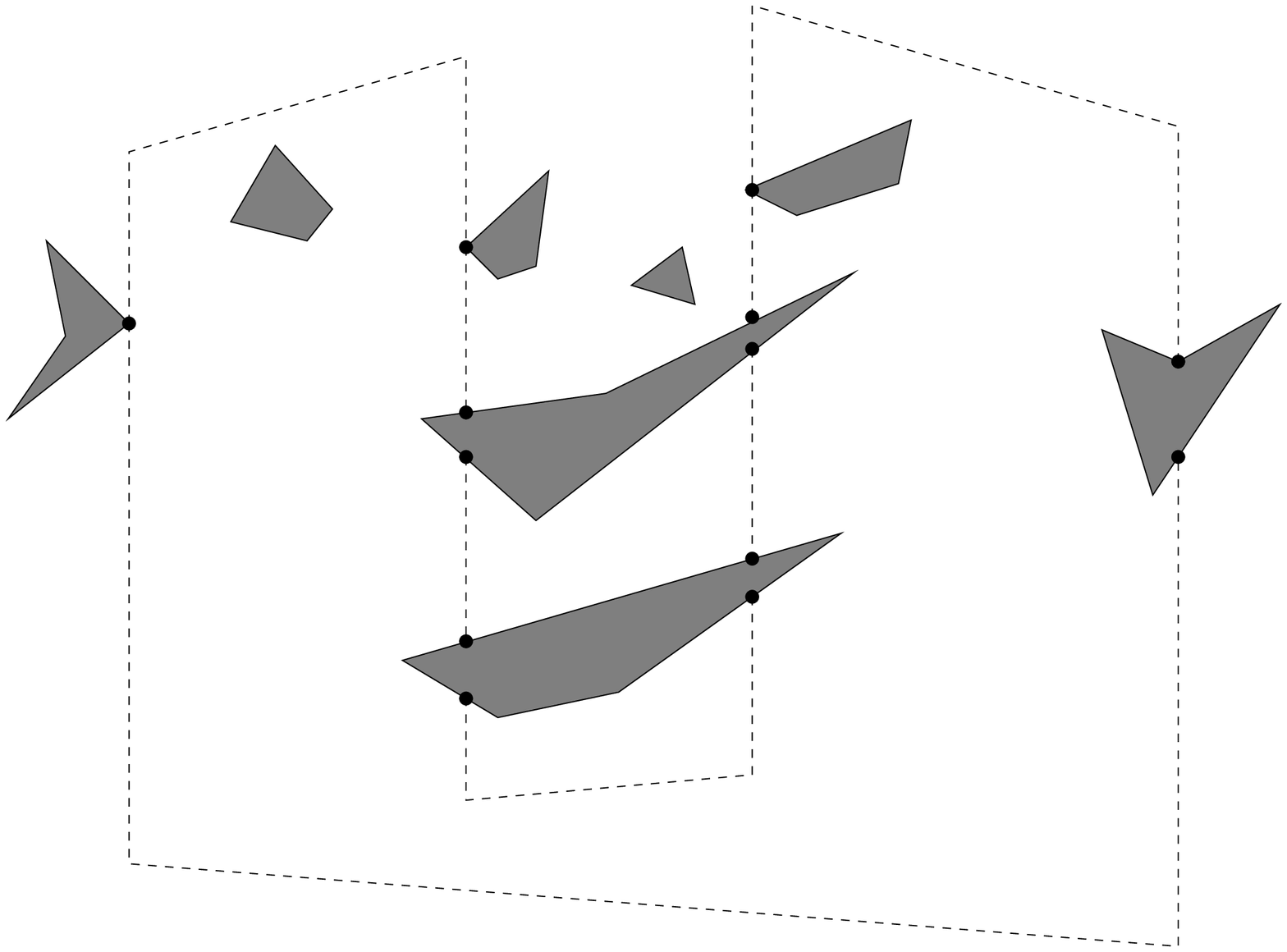}
\end{tabular}
\caption{(a) The separator $\Sigma$ which may pass through polygon holes, which
are shaded. (b) The fragments in $\Sigma \cap \polygon$.}
\label{fig:step2-1}
\end{figure}

We slightly modify fragment $\sigma$ so that both its endpoints are vertices of $\polygon$: if an
endpoint of $\sigma$ is a point $p$ on the interior of edge $f$ of the polygon $\polygon$, we extend
$\sigma$  by adding the segment from $p$ to the {\em left} endpoint of $f$. See
Figure~\ref{fig:step2-2}. Note that $n(\sigma,C)$
remains unchanged as a consequence of this. The set $\Sigma'$ of fragments now satisfies the following
properties:

\begin{enumerate}
\item $\sum_{\sigma \in \Sigma'}n(\sigma) \leq \sum_{e \in \Sigma} n(e) \leq \delta K/30$.
\item Each fragment $\sigma \in \Sigma'$ begins and ends at a vertex of $\polygon$ and contains
      no vertex of $\polygon$ in its relative interior.
\item No fragment in $\Sigma'$ intersects the interior of any convex polygon in $\Cint \cup \Cext$.
\item The fragments in $\Sigma'$ partition 
      $\polygon$ into connected components with the property that no component contains a    
      polygon from $\Cint$ as well as a polygon from $\Cext$. (That is, a 
component may contain polygons from $\Cint$, or polygons from $\Cext$, but
not polygons from both $\Cint$ and $\Cext$.)
\item Each fragment $\sigma \in \Sigma'$ is not self-intersecting. However, the two
      endpoints of a fragment may be the same point.
\item No two fragments in $\Sigma'$ cross.
\end{enumerate}

\begin{figure}[hbt]
\centering
\begin{tabular}{c@{\hspace{0.1\linewidth}}c}
\includegraphics[scale=0.3]{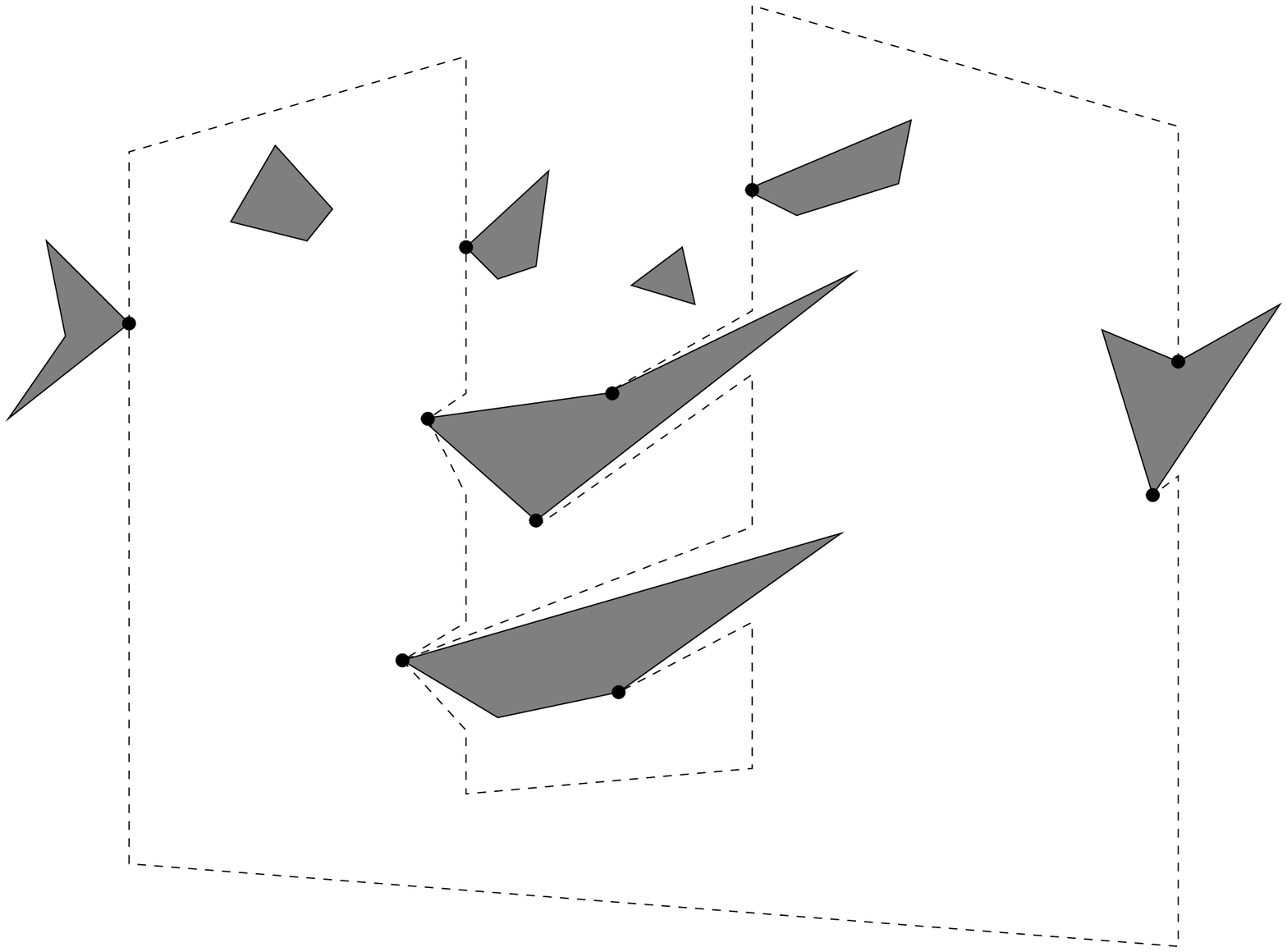} &
\includegraphics[scale=0.3]{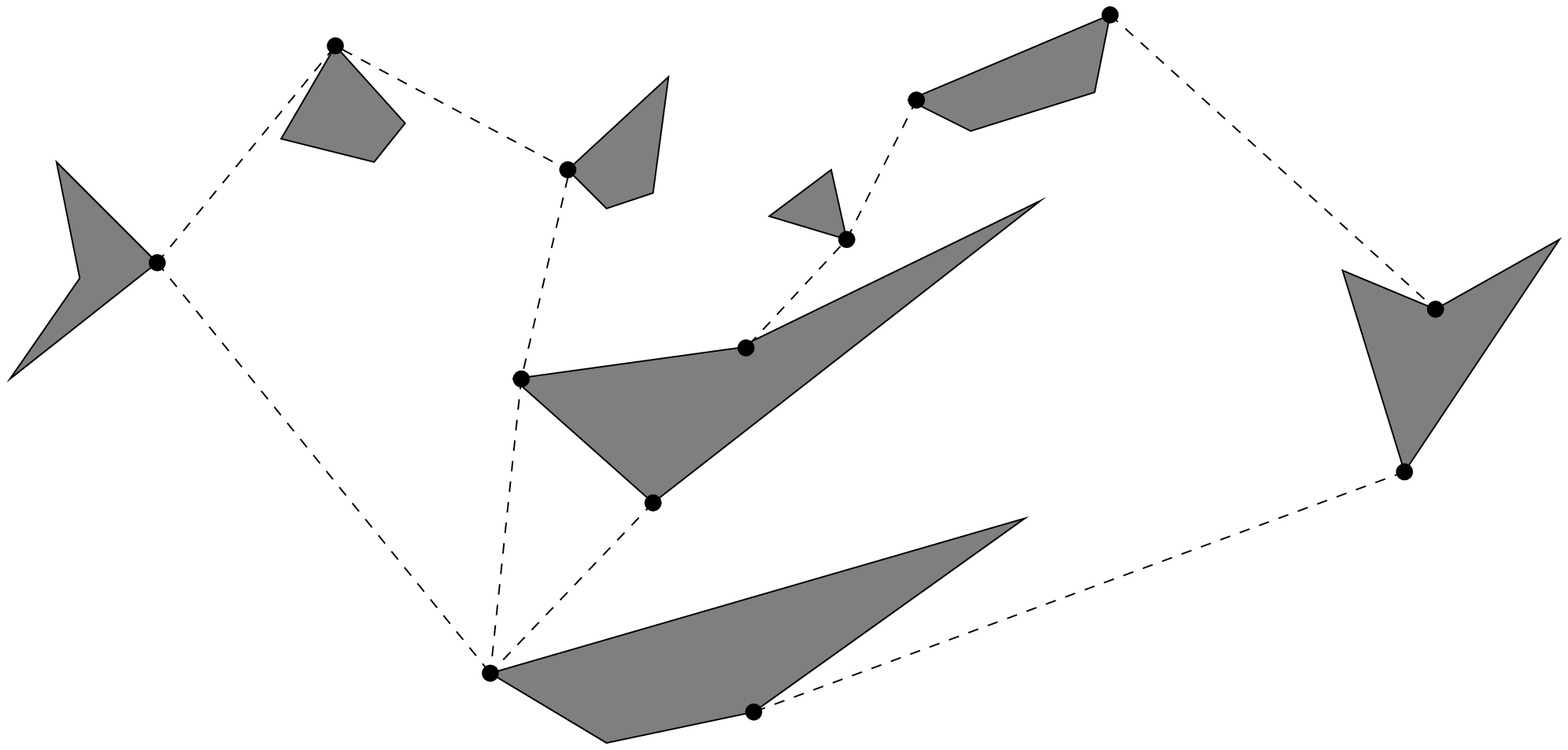}
\end{tabular}
\caption{(a) Modifying each fragment so that endpoints are polygon vertices. (b) The diagonals resulting from replacing each fragment by a shortest homotopic 
path.}
\label{fig:step2-2}
\end{figure}
   
For each $\sigma \in \Sigma'$, we compute the shortest path $\overline{\sigma}$ in $\polygon$ between the 
endpoints of $\sigma$ that is {\em homotopic to} $\sigma$. This can be done
efficiently \cite{hersh,efrat,cabello,besp}. Let us denote by
$n(\barsigma,C)$ the number of connected components of $\barsigma \cap (\text{interior } \ C)$, for any $C \in \C$.
Then we have $n(\barsigma,C) \leq n(\sigma,C)$ for any polygon $C \in \C$. In particular, if the interior
of $C$ is not intersected by $\sigma$ then it is not intersected by $\overline{\sigma}$ as well. Let
$n(\barsigma) = \sum_{C \in \C} n(\barsigma,C)$.

Let $\barSigma = \{ \barsigma \ | \ \sigma \in \Sigma'$\}. The fragments in $\barSigma$ satisfy
the following properties:

 \begin{enumerate}
\item $\sum_{\barsigma \in \barSigma}n(\barsigma) \leq \sum_{e \in \Sigma} n(e) \leq \delta K/30$.
\item Each fragment $\barsigma \in \barSigma$ begins and ends at a vertex of $\polygon$ and contains
      no vertex of $\polygon$ in its relative interior.
\item No fragment in $\barSigma$ intersects the interior of any convex polygon in $\Cint \cup \Cext$.
\item The removal of the points corresponding to the fragments in $\barSigma$ partitions 
      $\polygon$ into connected components with the property that no component contains a    
      polygon from $\Cint$ as well as a polygon from $\Cext$.
\item Each fragment $\barsigma \in \barSigma$ is not self-intersecting. However, the two
      endpoints of a fragment may be the same point.
\item No two fragments in $\barSigma$ cross.
\end{enumerate}  

Note that each fragment $\barsigma$, being a shortest homotopic path, is constituted of a sequence of diagonals and edges from $\polygon$. So
we define $D(\Sigma)$ as the set of diagonals corresponding to the fragments in $\barSigma$. (A diagonal
in $D(\Sigma)$ can be present in more than one fragment of $\barSigma$.)
See Figure~\ref{fig:step2-2}.

The last two fragment properties of $\barSigma$ imply that $D(\Sigma)$ is a conforming set of diagonals. Notice that
the number of diagonals in $D(\Sigma)$ can be much greater than the number of edges in $\Sigma$. 
However, $D(\Sigma)$ is uniquely and efficiently computed given $\Sigma$. Since $\Sigma$ comes
from a family of $n^{O(1/\delta^2)}$ cycles that can be computed in $n^{O(1/\delta^2)}$ time given
$\polygon$, $D(\Sigma)$ comes from a family of  $n^{O(1/\delta^2)}$ diagonal subsets that can be
computed in $n^{O(1/\delta^2)}$ time given $\polygon$.

\paragraph{Step 3:} For a diagonal $d \in D(\Sigma)$, let $n(d,C) = 1$ if $d$ intersects the interior
of $C \in \C$, and $0$ otherwise. Let $n(d) = \sum_{C \in \C}n(d,C)$. The first property of
$\barSigma$ can be restated as saying that $\sum_{d \in D(\Sigma)} n(d) \leq \delta K / 30$. 

The diagonals in $D(\Sigma)$ partition $\polygon$ into a set of smaller polygons $\{P_1, P_2,
\ldots, P_s\}$. We now show
how to obtain, from $\C$, convex decompositions of these smaller polygons that obey the size
bounds claimed in the lemma. We will think of these new convex decompositions as a 
new convex decomposition of $\polygon$ that respects the set $D(\Sigma)$ of diagonals. The
new convex decomposition will have the convex polygons in $\Cint$ and $\Cext$ -- the interiors
of these polygons do not intersect the diagonals in $D(\Sigma)$. Let $\Cbad = \C \setminus 
\{ \Cint \cup \Cext \}$. From the properties of $\Sigma$, it follows that $|\Cbad| \leq
\delta K/30$. We show below that we can obtain a convex decomposition of size at most $\delta K$ for
the portion of $\polygon$ that is covered by the polygons in $\Cbad$. This convex decomposition will
respect the set $D(\Sigma)$. Since smaller polygon $P_j$ does not have a polygon from both $\Cint$ and
$\Cext$,  it follows that 
\[K(P_j) \leq \max \{ |\Cint|, |\Cext|\} + \delta K \leq (2/3 + \delta) K.\] It also follows that
\[ \sum_j K(P_j) \leq |\Cint| + |\Cext| + \delta K \leq (1 + \delta) K.\]

We describe the construction of $\Cnew$, the new convex decomposition of the portion of $\polygon$ that 
is covered by the polygons in $\Cbad$. This $\Cnew$ respects the diagonals in $D(\Sigma)$, that is,
the interior of no convex polygon in $\Cnew$ is intersected by a diagonal in $D(\Sigma)$. For each
convex polygon $C \in \Cbad$, consider the subset $D(C) \subseteq D(\Sigma)$ of diagonals that
intersect the interior of $C$. Let $V(C)$ denote those vertices of $C$ that do not lie on any
diagonal in $D(\Sigma)$. Define the following relation on $V(C)$: $u$ and $v$ are related if the line
segment joining them does not intersect any diagonal in $D(C)$. It is easy to see that this is an
equivalence relation. Let $V_1, V_2, \ldots, V_m$ be the equivalence classes. It is easy to see that
$m \leq D(C) + 1 $. We add $\conv(V_i)$ to $\Cnew$  if $\conv(V_i)$ is a 2-dimensional object, that is,
not a line segment or a point. See Figure ~\ref{fig:step31}. The number of convex polygons contributed by $C$ to $\Cnew$ is 
at most $1 + D(C)$, so the number of convex polygons in $\Cnew$ overall is at most
\[ \sum_{C \in \Cbad} (1 + D(C)) \leq |\Cbad| + \sum_{d \in D(\Sigma)} n(d) \leq \delta K/15.\]

\begin{figure}[hbt]
\centering
\begin{tabular}{c@{\hspace{0.1\linewidth}}c}
\includegraphics[scale=0.5]{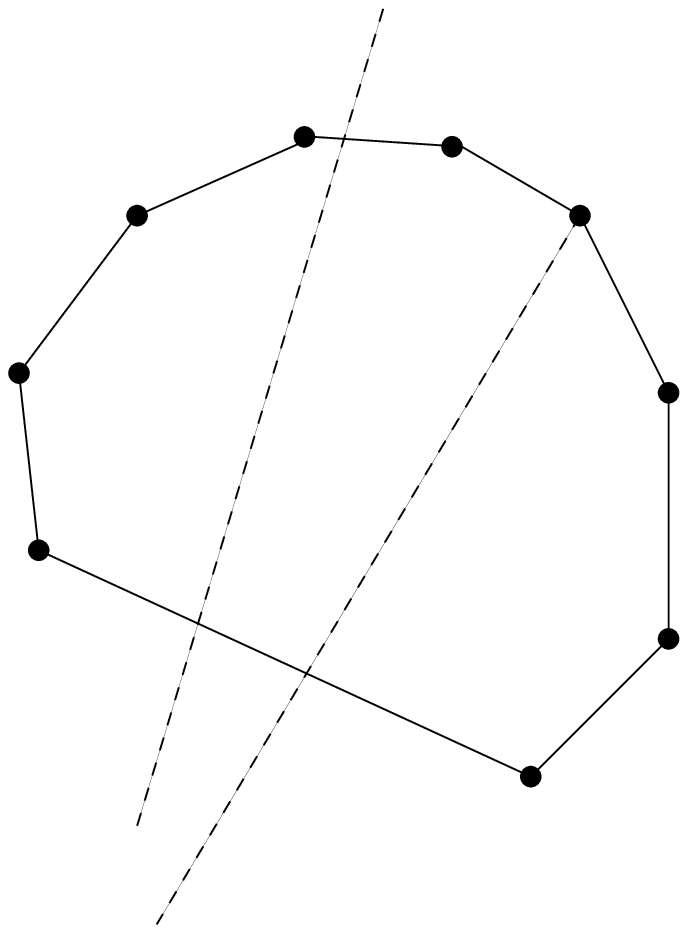} &
\includegraphics[scale=0.5]{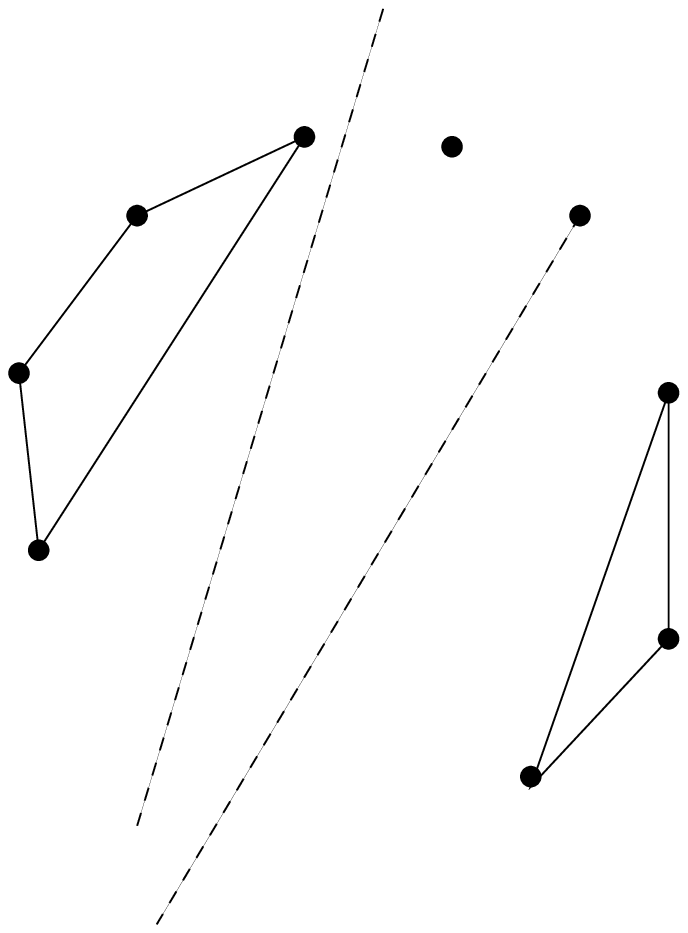}
\end{tabular}
\caption{(a) A convex polygon $C \in \Cbad$, and the diagonals in $D(\Sigma)$ that
intersect it. (b) The convex polygons added to $\Cnew$ from $C$.}
\label{fig:step31}
\end{figure}

For each polygon $\polygon_j$ in the partition of $\polygon$ induced by $D(\Sigma)$, consider the portion that is outside the polygons of $\Cint$, $\Cext$, and $\Cnew$. This portion is a set of polygons. We triangulate each such polygon,
and add the resulting triangles to $\Cnew$. See Figure \ref{fig:step32}. Note
that triangulating a polygon with $m$ vertices results in at most $3m$ triangles, even if the polygon has holes. This
completes the construction of $\Cnew$.

We need to bound the number of triangles added in this step, summed over
all $\polygon_j$. To this end,
let $\lambda$ denote the sum of the number of vertices of all the polygons
we triangulate. To bound $\lambda$, we observe that each $C \in \Cbad$
``contributes'' at most $8 D(C)$ vertex-polygon features to $\lambda$.  
Thus, $\lambda \leq 8 \sum_{C \in \Cbad} D(C) \leq 8 \delta K/30$.
So the number of triangles we add to $\Cnew$ is at most $3 \lambda \leq
24 \delta K/30$. 

Thus, $|\Cnew| \leq \frac{24 \delta K}{30} + \frac{\delta K}{15} \leq \delta K$. We now have the desired convex decomposition of
$\polygon$ that respects $D(\Sigma)$: $\Cint \cup \Cext \cup \Cnew$. This completes the proof of the lemma.

\begin{figure}[H]
\centering
\begin{tabular}{c@{\hspace{0.1\linewidth}}c}
\includegraphics[scale=0.3]{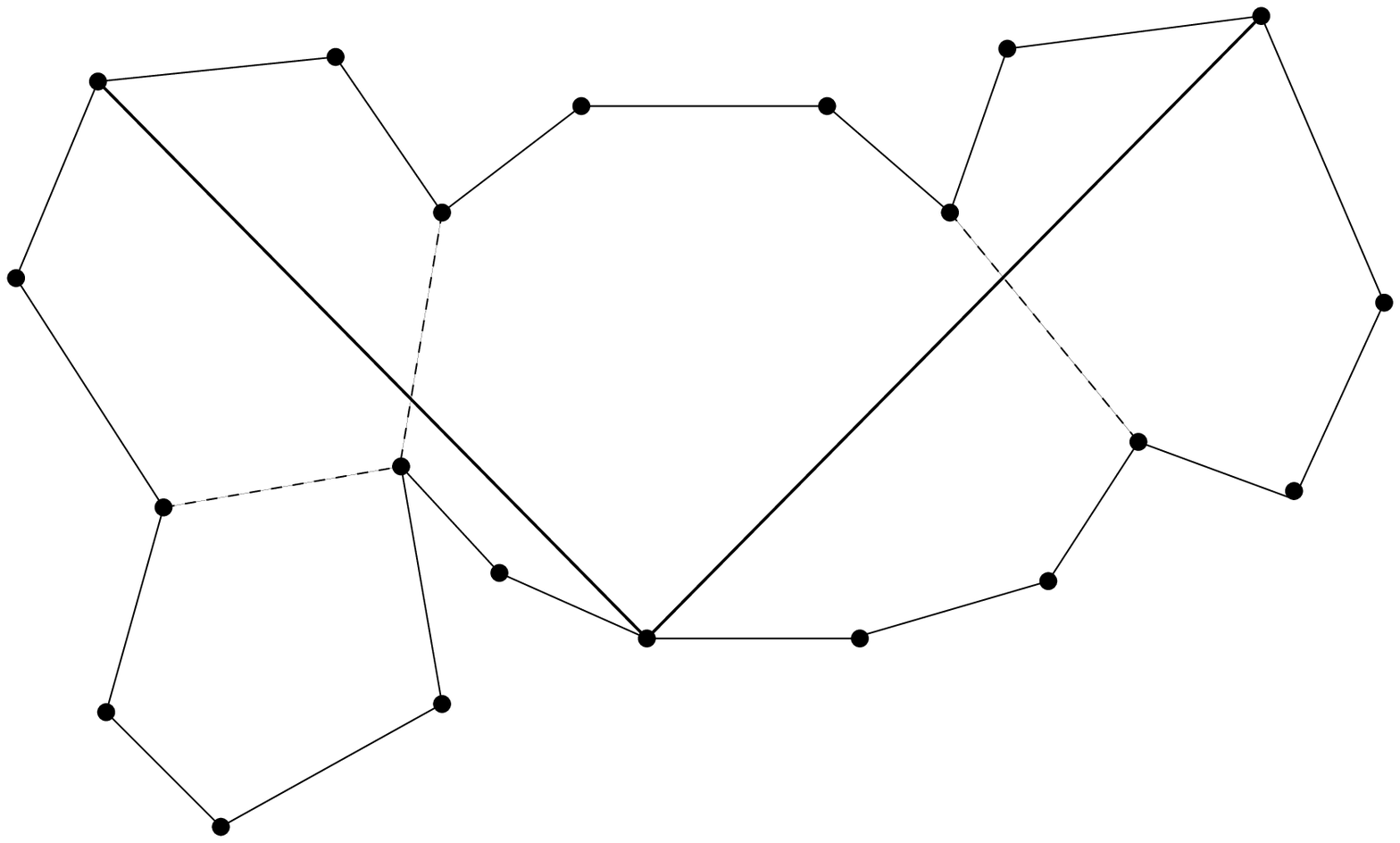} &
\includegraphics[scale=0.3]{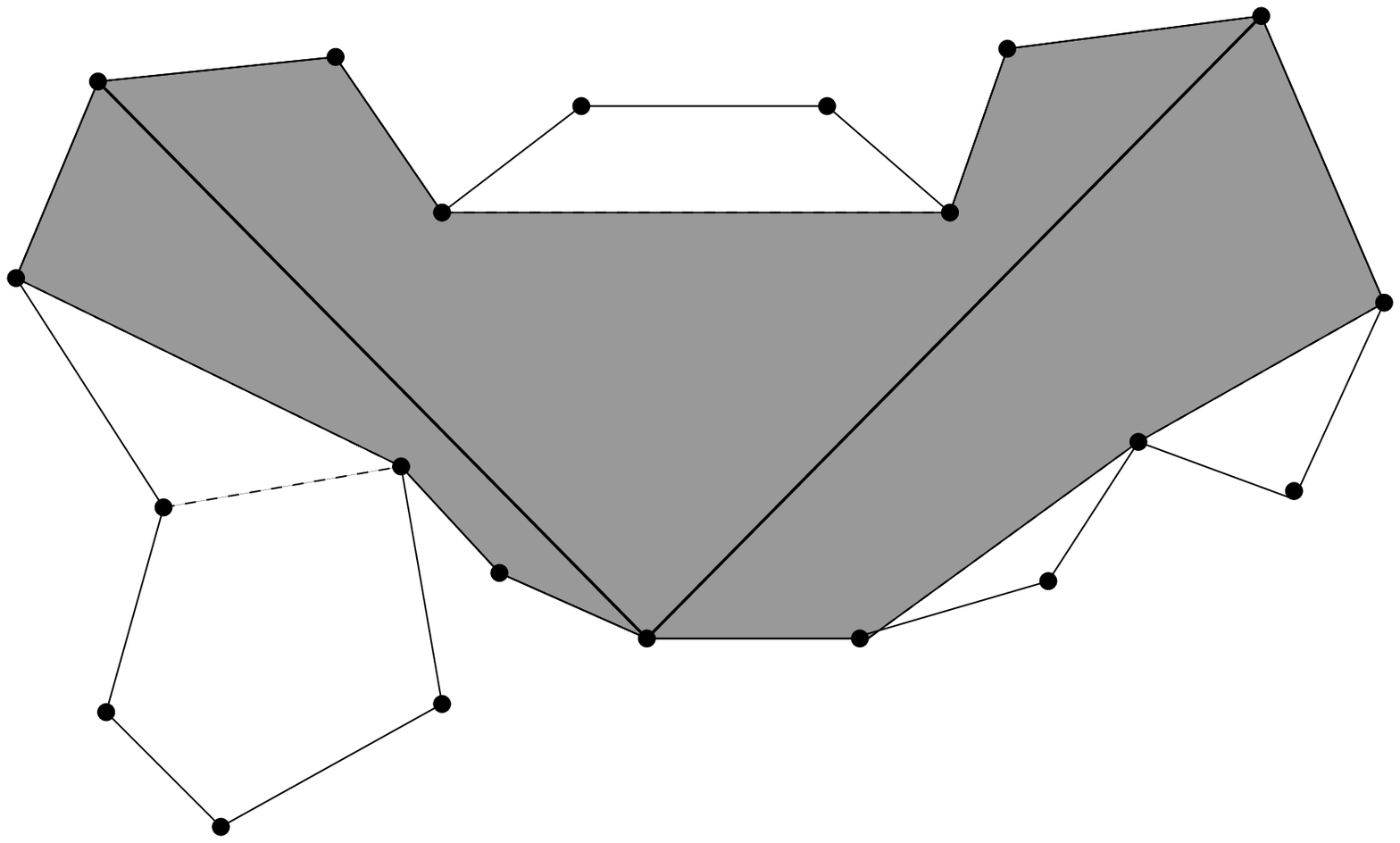}
\end{tabular}
\caption{(a) An illustration of a polygon $P$, an optimal decomposition $\C$
using (dashed) diagonals, and the two diagonals in $D(\Sigma)$ (bold). (b) 
The three polygons in $\Cbad$ contribute a total of four convex polygons to
$\Cnew$; in addition, we triangulate the shaded polygons and add the triangles
to $\Cnew$.}
\label{fig:step32}
\end{figure}

\subsection{Algorithmic Aspects}

We now use the Lemma~\ref{lem:sep1} to develop a QPTAS for the convex
decomposition problem. For this purpose, we need a good exact algorithm
to serve as the base case for our recursive algorithm. Suppose $P'$ is
an $n$-vertex polygon, and the optimal decomposition for it has 
$K = K(P')$ convex polygons $P'_1, \ldots, P'_k$. We argue that the number of 
diagonals added is at most $3K - 6$. To see this, construct a graph where there
is a vertex for each $P'_i$, and an edge for each diagonal between the (vertices
 corresponding to the ) two convex polygons it is incident to. This graph is
clearly planar. Furthermore, since the $P'_i$ are convex, the graph has no parallel edges. As such, the number of edges, and hence diagonals, is at most
$3K - 6$.

Thus, given an $n$-vertex polygon $P'$ and a $k \geq 0$, we can check if
$P'$ admits a convex decomposition of size at most $k$ in $n^{O(k)}$ time. We only need to try all conforming subsets of at most $3k - 6$ diagonals. In the
same time bound, we can find an optimal convex decomposition, assuming it
has size at most $k$.

We now describe our QPTAS. It will be convenient to describe a 
non-deterministic algorithm first. Assuming it makes the right separator 
choices, we can analyze the approximation guarantee. Subsequently, we
make the algorithm deterministic and bound its running time.

\paragraph{Nondeterministic Algorithm.}
Our algorithm $\decompose(P')$ takes as input a polygon $P'$ and returns 
a decomposition of $P'$. It uses a parameter $0 < \delta < \frac{3}{4} - 
\frac{2}{3}$ that we specify later. Let $\lambda = \frac{ c \log (1/\delta)}{\delta^3}$,
the threshold in Lemma \ref{lem:sep1}. Since $P'$ will be a subpolygon of $P$,
the number of its vertices is at most $n$. Our overall algorithm simply
invokes $\decompose(P)$.

\begin{enumerate}
\item We check if $P'$ has a decomposition with at most $\lambda$ convex polygons. If so, we return the optimal decomposition. This is the base case of our
algorithm. This computation can be done as described
above in $n^{O(\lambda)}$ time. Henceforth, we assume that $K(P') > \lambda$.

\item Compute the family  $\D = \{ D_1, D_2, \ldots, D_t \}$ of sets of diagonals, as stated in Lemma \ref{lem:sep1}, for $P'$.

\item Choose a $D_i \in \D$.

\item Suppose $D_i$ partitions $P'$ into subpolygons $P'_1,P'_2,\ldots,P'_s$.
      Return 
      \[ \bigcup_{j=1}^s \decompose(P'_j).\]
\end{enumerate}

\paragraph{Approximation Ratio.}
We define the level of a polygon $P'$ to be the integer $i > 0$ such that
$\lambda (4/3)^{i-1} < K(P') \leq \lambda (4/3)^i$. If $K(P') \leq \lambda$,
we define its level to be $0$. Thus if $\decompose(P')$ is solved via the
base case, then the level of $P'$ is $0$. The following lemma bounds the 
quality of approximation of our non-deterministic algorithm.

\begin{lemma}
\label{lem:approx1}
Assume that $\delta < \frac{3}{4} - \frac{2}{3}$. There is an instantiation
of the non-deterministic choices for which $\decompose(P')$ returns a
convex decomposition with at most $(1 + \delta)^{\ell} K(P')$ polygons, where $\ell$
is the level of $P'$.
\end{lemma}

\begin{proof}
The proof is by induction on $\ell$. The base case is when $\ell = 0$, and here the statement follows from the base case of the algorithm. So assume that $\ell > 1$, and
that the statement holds for instances with level at most $\ell - 1$.

Suppose that the algorithm non-deterministcally picks the $D_i \in \D$ that
satisfies the guarantees of Lemma \ref{lem:sep1} for $P'$. Let $P'_1, P'_2, \ldots, P'_s$ be the subpolygons that result from partitioning $P'$ with $D_i$. 

Since $K(P'_j) \leq (2/3 + \delta) K(P') \leq (3/4) K(P')$, it follows that the
level of each $P'_j$ is at most $\ell - 1$. Thus, for each $j$, there are
nondeterministic choices for which $\decompose(P'_j)$ returns a decomposition
of $P'_j$ with at most $(1 + \delta)^{\ell -1} K(P'_j)$ polygons. Thus, the
size of the decomposition of $P'$ returned by $\decompose(P')$ is at most
\[(1 + \delta)^{\ell -1} \sum_j  K(P'_j) \leq (1 + \delta)^{\ell} K(P').\]
\end{proof}

\paragraph{Deterministic Algorithm.} Since a triangulation of $P$, the original
input polygon, uses at most $3n - 6$ triangles, the level of $P$ is at most
$\alpha = \log_{4/3}(3n - 6)$. It
follows that with $\decompose(P)$, for suitable non-deterministic 
separator choices, returns a decomposition with at most $(1 + \delta)^
\alpha$ times the size of the optimal disjoint cover. Furthermore, the depth of the recursion with
such seperator choices is at most $\alpha$.

To get a deterministic algorithm, we make the following natural changes
to $\decompose(P')$.
If a call to $\decompose(P')$ is at recursion depth that is greater
than $\alpha$ (with respect to the root corresponding to
$\decompose(P)$), we return a special symbol $I$.  
In the $\decompose(P')$ routine, when we are not in the base case, 
we try all possible separators $D_i \in \D$ instead of nondeterministically 
guessing one -- we return the smallest sized set 
$\bigcup_{j=1}^s \decompose(P'_j)$, over all $i$ for which none of the recursive calls $\decompose(P'_j)$ returns $I$. If no such $i$ exists, 
$\decompose(P')$ returns $I$.

With these changes, $\decompose(P)$ is now a deterministic algorithm
that returns a decomposition of size at most $(1 + \delta)^{\alpha} K(P)$. 
Its running time is 
\[ \left( n^{O(1/\delta^2)} \right)^{\alpha} \cdot
   n^{O(\lambda)} = n^{O\left( (\log n + \log 1/\delta)/\delta^3 \right)} .\]

Plugging $\delta = \eps/2 \alpha$, the approximation
guarantee is $(1 + \eps)$ and the running time is
$n^{O((\log n/\eps)^4)}$. We can thus conclude with our main result for
convex decomposition:

\begin{theorem}
There is an algorithm that, given a polygon $P$ and an $\eps > 0$, runs in time
$n^{O((\log n/\eps)^4)}$ and returns a diagonal-based convex decomposition of $P$
with at most $(1 + \eps) K(P)$ polygons, where $K(P)$ is the number of polygons
in an optimal diagonal-based convex decomposition of $P$. Here $n$ stands for
the number of vertices in $P$.
\end{theorem}

\section{Surface Approximation}
We now describe our algorithm for the surface approximation problem. Recall that we are
given a set $\Sb$ of $n$ points in $\Real^3$ sampled from a bi-variate
function $f(x, y)$, and another parameter $\mu > 0$. A piece-wise linear function $g(x,y)$ is an approximation of $f(x,y)$ if $\forall \pb = (x,y,z) \in \Sb, |g(x,y) - z| \leq \mu$. The bi-variate function $f(x,y)$ represents the 
surface from which the points are sampled, and we want to compute an \textit{approximate} polyhedral surface $g(x,y)$ with minimal complexity. The complexity
of a piecewise linear surface is defined to be the number of its faces, which are required to be triangles. 

For any point $\pb \in \Sb$ which is in
$\Real^3$, we define $\p$ to be the projection of $\pb$ on to the
$xy$-plane. Let $S = \{p\ |\ \bar{p} \in \Sb\}$. A triangle $\triangle$ in
the $xy$-plane is a \textit{valid} triangle if it is the projection of a
triangle $\overline{\triangle}$ in $\Real^3$, such that $\forall \p \in S
\cap \triangle$ the vertical distance between $\overline{\triangle}$ and $\pb$
is at most $\mu$. Agarwal and Suri \cite{surface_approx_suri} have shown that the surface
approximation problem is equivalent, up to multiplicative constant factors,
to computing a minimum-cardinality cover for $S$ using a set of valid triangles with 
pairwise-disjoint interiors. Notice that the set of valid triangles can be
infinite. We describe a method for computing a polynomial-sized set $\mathcal{B}$ of 
valid triangles, termed the \textit{basis}, such that  the surface
approximation problem is equivalent, up to multiplicative constant factors,
to computing a minimum-cardinality cover for $S$ using a subset of {\em basis triangles} with 
pairwise-disjoint interiors. As we then show, the basis triangles have a certain
closure property that enables us to obtain an approximation scheme
for the above covering problem using the seperator approach. 

\subsection{Construction of the basis}
Let $\mathcal{T}$ be the set of all valid triangles in the plane, which can be
infinite. Let $\F = \{ S \cap \triangle \ | \ \triangle \mbox{ is a triangle}\}$.
It is easy to see that set $\F$ has size $O(n^6)$, and can be computed in, say,
$O(n^7)$ time.

\begin{figure}[hbt]
\tikzstyle{node} = [circle, fill=blue, minimum size=4pt, inner sep=0pt]
\centering
  \begin{subfigure}[b]{0.3\textwidth}
  \centering
  \resizebox{\linewidth}{!} {   
   \begin{tikzpicture}
       \coordinate (a) at (1.75,1.77);
       \coordinate (p) at (3.6,0.44);
       \coordinate (c) at (5.46,1.07);
       \coordinate (q) at (6.06,2.68);
       \coordinate (e) at (3.38,4.02);
       \coordinate (f) at (3.98,2.8);
       \coordinate (g) at (2.78,2.48);
       \coordinate (h) at (5.3,4.14);
       \coordinate (r) at (2.14,3.32) ;
       \coordinate (j) at (2.22,0.61);
       \coordinate (k) at (4,8);
       \coordinate (l) at (0,0);
       \coordinate (m) at (8,0);
       \coordinate (n) at (3.52,1.71);
       \coordinate (o) at (4.91,4.89);
       \coordinate (t) at (1.48,0.7);
       
       \draw (k) -- (l) -- (m) -- (k);
       \foreach \name in {a,c,e,f,g,h,j,n}
       {
          \node[node, label =$ $] at (\name) {$ $};
       }
       \foreach \name in {p,q,r}
       {
          \node[node, label =$\name$] at (\name) {$ $};
       }
      \draw[dashed] (e) -- (r) -- (a) -- (j) -- (p) -- (c) -- (q) -- (h) -- (e);
   \end{tikzpicture}
   }
   \end{subfigure}
   \hfill
   \begin{subfigure}[b]{0.3\textwidth}
    \centering
  \resizebox{\linewidth}{!} {   
   \begin{tikzpicture}
   \tikzset{hexagon/.style={color = red, thick, densely dotted}}
       \foreach \name in {a,c,e,f,g,h,j,n}
       {
          \node[node, label =$ $] at (\name) {$ $};
       }
       \foreach \name in {p,q,r}
       {
          \node[node, label =$\name$] at (\name) {$ $};
       }
       \draw[hexagon] (p) -- ($(p)!3.5cm!(c)$);
       \draw[hexagon] (p) -- ($(p)!3.5cm!(j)$);
       \draw[hexagon] (q) -- ($(q)!2.75cm!(c)$);
       \draw[hexagon] (q) -- ($(q)!4cm!(h)$);
       \draw[hexagon] (r) -- ($(r)!4cm!(e)$);
       \draw[hexagon] (r) -- ($(r)!3.25cm!(a)$);
       
       \draw[opacity=1] (k) -- (l) -- (m) -- (k);
   \end{tikzpicture}
   }
   \end{subfigure}
   \hfill
   \begin{subfigure}[b]{0.3\textwidth}
    \centering
  \resizebox{\linewidth}{!} {   
   \begin{tikzpicture}
       \foreach \name in {a,c,e,f,g,h,j,n}
       {
          \node[node, label =$ $] at (\name) {$ $};
       }
       \foreach \name in {p,q,r}
       {
          \node[node, label =$\name$] at (\name) {$ $};
       }
       \draw (k) -- (l) -- (m) -- (k);

       \draw (p) -- (r);
       \draw (p) -- (o);
       \draw (p) -- (q);   
       
       \draw (p) -- (c) -- (q) -- (o) -- (r) -- (t) -- (p);
   \end{tikzpicture}
   }
   \end{subfigure}
   \vspace{0.5cm}
   \caption{ (a) An arbitrary triangle $\triangle$ and the set $R = \triangle \cap S \in \mathcal{F}$. The dashed polygon is $\conv(R)$. (b) Three points $p,q,r$ on $\text{Conv}(R)$ and the corresponding hexagon $H_{pqr}$ formed by the edges of $\text{Conv}(R)$ incident on $p,q,r$. (c) Triangulation of $H_{pqr}$ results in the addition of at most 4 triangles to $\mathcal{B}$, which cover the points in $R$.
}
\label{fig:basis}
\end{figure}
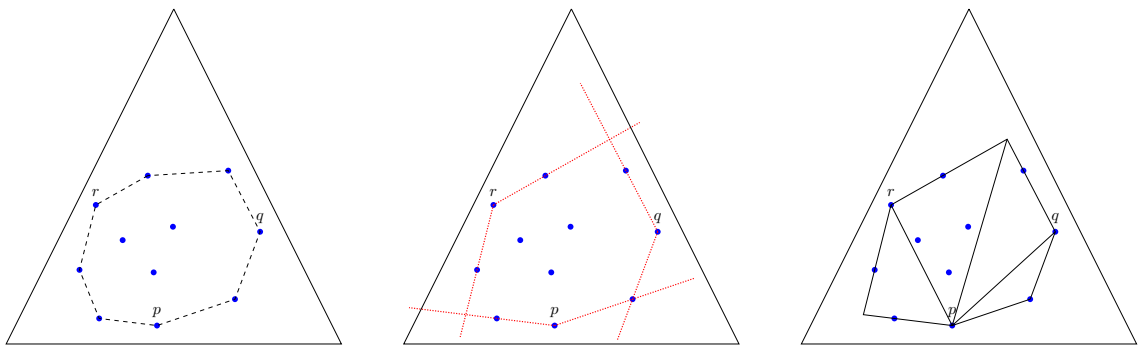

For each $R \in \F$, we compute the convex hull $\conv(R)$ of the points in $R$. 
If $\conv(R)$ consists of a point, or a single edge, we add the degenerate 
triangle $\conv(R)$ to the basis $\mathcal{B}$. Otherwise, $\conv(R)$ is
2-dimensional and has at least three vertices. For each triple $\{p,q,r\}$
of vertices in $\conv(R)$, we contruct the hexagon $H_{pqr}$ formed by the edges of $\conv(R)$ 
incident on $p$, $q$, and $r$. $H_{pqr}$ may be degenerate i.e.\ it may not be a hexagon, or it may
be unbounded. In case $H_{pqr}$ is bounded, we triangulate the hexagon by using diagonals
from the bottom vertex, and add the resulting set of at most $4$ triangles to 
$\mathcal{B}$. See Figure~\ref{fig:basis}

Since we generate at most $O(n^3)$ hexagons from $\conv(R)$, and from each such hexagon
we generate at most $4$ triangles, the basis $\mathcal{B}$ would now consist of at most
$O(n^9)$ triangles.

\paragraph{Filtering the basis:} We remove any triangle ${\triangle} \in
\mathcal{B}$ that is not a valid triangle. This can be done by solving a simple $3$-dimensional
linear program for each triangle in $\mathcal{B}$, as shown by Agarwal and Desikan \cite{surface_approx_desikan}. 
Let $S_{\triangle} = S \cap {\triangle}$ be the set of points contained inside ${\triangle}
\in \mathcal{B}$. Since ${\triangle}$ is a valid triangle,
then there would exist a triangle $\overline{\triangle}$ in $\Real^3$ such
that ${\triangle}$ is the projection of $\overline{\triangle}$ on the
$xy$-plane, and the vertical distance between $\overline{\triangle}$ and any
point in $\{\bar{p}\ |\ p \in S_{\triangle} \}$ is
at most $\mu$. This completes the description of the basis computation.

A useful property of the basis is summarized below.

\begin{lemma} Let $\triangle$ be any valid triangle. There exist a set $\mathcal{B}(\triangle)
\subseteq \mathcal{B}$ of at most
four triangles, such that (a) the triangles in $\mathcal{B}(\triangle)$ have
pair-wise disjoint interiors; (b) each of the triangles in $\mathcal{B}(\triangle)$ is contained
in $\triangle$; and (c) $\mathcal{B}(\triangle)$  covers $S \cap \triangle$.
\end{lemma}
\begin{proof} 
   Let $R = S \cap {\triangle}$. If $\conv(R)$ is $0$- or $1$-dimensional, we
   have added the degenerate triangle $\conv(R)$ itself to $\mathcal{B}$, and
   the lemma holds with $\mathcal{B}(\triangle) = \{\conv(R)\}$ . 
   Assume henceforth that $\conv(R)$ is $2$-dimensional.

Let $h_1, h_2 , h_3$ be
   the half-planes defined by the 3 edges of ${\triangle}$, such that $h_1
   \cap h_2 \cap h_3 = {\triangle}$. Let $p_i$ be the point in $R$
   that is closest to the line bounding $h_i$ -- if there is a tie, we break it arbitrarily.
   Consider the hexagon $H_{p_1p_2p_3}$ formed by extending the edges of $\conv(R)$
   incident to the $p_i$. Our procedure for generating the basis would have
   generated the hexagon  $H_{p_1p_2p_3}$ while considering $R$. It is not
   hard to see, as we explain below, that $H_{p_1p_2p_3} \subseteq \triangle$.
   The set of at most $4$ triangles that we obtain by triangulating $H_{p_1p_2p_3}$
   are added to $\mathcal{B}$. This set $\mathcal{B}(\triangle)$ of triangles has the properties claimed.

   We now show that $H_{p_1p_2p_3} \subseteq \triangle$. Let $W_i$ be the wedge
   whose apex is at $p_i$ and whose bounding rays are the ones containing the
   two edges of $\conv(R)$ incident at $p_i$. Since $p_i$ is the point in
   $\conv(R)$ that is closest to the line bounding $h_i$, it follows that
   the two rays bounding the edge $W_i$ do not contain any point outside
   $h_i$. That is, $W_i \subseteq h_i$. This implies that
\[ H_{p_1p_2p_3} = W_1 \cap W_2 \cap W_3 \subseteq h_1 \cap h_2 \cap h_3 = \triangle.\] 
\end{proof}

The next two observations relate the surface approximation problem to that
of computing a minimal cover of $S$ using a set of pairwise-disjoint triangles 
from the basis $\mathcal{B}$.  
A consequence of our basis if the following.

\begin{lemma}
There is a set of at most $4 OPT$ triangles from $\mathcal{B}$, with pairwise-disjoint
interiors, that covers $S$, where $OPT$ is the complexity of an optimal 
solution to our surface approximation instance.
\end{lemma}

\begin{proof}
Consider the set $\mathcal{T}' \subseteq \mathcal{T}$ of triangles that are formed by 
projecting the triangular faces in the optimal solution. The set 
$\bigcup_{\triangle \in \mathcal{T}'} \mathcal{B}(\triangle)$ has the properties
claimed.
\end{proof} 
  
The next observation is due to Agarwal and Suri \cite{surface_approx_suri}.

\begin{lemma}
If we have a cover of $S$ using $m$ pairwise-disjoint triangles from $\mathcal{T}$,
then we can efficiently compute a solution to the surface approximation problem with complexity
$O(m)$.
\end{lemma}

The above two lemmas imply that if we have an $O(1)$-approximation to the problem 
of computing a minimal cover of $S$ using a set of pairwise-disjoint triangles 
from the basis $\mathcal{B}$, then we have an $O(1)$-approximation for the
original surface approximation problem.

\subsection{A Disjoint Cover Using Basis Triangles}
We now describe a QPTAS for the problem of computing the smallest 
pair-wise disjoint subset of $\B$ that covers $S$. 

\paragraph{The Separator.} 
We need the following separator computation, which is very similar 
to the constructions in \cite{AdamaszekW14,Har-Peled13,MRS14} and Step 1 of the separator theorem
for convex decomposition. Our separators will be closed, simple,
polygonal curves. For an edge $e$ on such a curve $C$, and for
a pairwise-disjoint subset $\D \subseteq \B$, let $n(e,\D)$ denote the
number of triangles in $\D$ whose relative interior is intersected by
$e$, and let $n(C,\D)$ denote $\sum_e n(e, \D)$, where the summation is
over all edges $e$ of $C$. 

\begin{lemma}
\label{lem:sep2}
Given $\B$, and $0 < \delta < 1$, we can compute in time $n^{O(1/\delta^2)}$
a family $\C = \{C_1, C_2, \ldots, C_t \}$ of closed, simple, polygonal 
curves, each with $O(1/\delta^2)$ vertices, with the following property:
for any subset $\D \subseteq \B$ with pairwise-disjoint triangles
such that $K := |\D| \geq \lambda := \frac{c\log 1/\delta}{\delta^3} $, there is a
$C_j \in \C$ such that (a) $n(C_j,\D) \leq \delta K/10$; (b) the number of triangles of $\D$ inside $C_j$ is at most $\frac{2K}{3}$; and (c) the number of
triangles of $\D$ outside $C_j$ is at most $\frac{2K}{3}$.
\end{lemma}  

\paragraph{The Algorithm.}
We describe a recursive procedure $\compcover(S',\B')$ that given as input subsets
$S' \subseteq S$ and $\B' \subseteq \B$, returns a cover of $S'$ with a
set of pairwise disjoint triangles from $\B'$. We assume that $\B'$ has
the following closure property: if $\triangle \in \B'$, and $\triangle_1 \in \B$ is contained in $\triangle$, then $\triangle_1 \in \B'$ as well. 
Our final algorithm simply
invokes $\compcover(S,\B)$. Our algorithm $\compcover(S',\B')$ is non-deterministic in the step
in which it makes a choice of separator. After analyzing the quality of
the solution produced by this non-deterministic algorithm, for suitable
separator choices, we discuss how it can be made deterministic.

\begin{enumerate}
\item By exhaustive search, we check if there is a subset of $\B'$ with at most
$\lambda = c \frac{\log 1/\delta}{\delta^3}$ pairwise disjoint triangles that covers
$S'$. If so, we return such a subset with minimum cardinality. This is the
base case of our algorithm. Henceforth, we assume that a minimal pairwise-disjoint cover needs at least $\lambda = \frac{c\log 1/\delta}{\delta^3}$ triangles.

\item Choose a separator $C_j \in \C$.

\item Let $S'_{j1}$ be the set of points in $S'$ that are inside $C_j$, and
      let $S'_{j2}$ be the remaining points in $S'$. Let $\B'_{j1}$ denote those
      triangles of $\B'$ that are inside $C_j$, and $\B'_{j2}$ denote those 
      triangles of $\B'$ that are outside $C_j$.

\item Return $\compcover(S'_{j1},\B'_{j1}) \cup \compcover(S'_{j2},\B'_{j2})$.
\end{enumerate}

We note that  $\B'_{j1}$ and $\B'_{j2}$ satisfy the closure property that
$\B'$ has.   

\paragraph{Approximation Ratio.}

Consider an input $(S',\B')$ to our algorithm, and suppose $\D' \subseteq 
\B'$ is a smallest pairwise-disjoint subset of $\B'$ that covers $S'$. We
define the {\em level} of the instance  $(S',\B')$ to be the integer $i > 0$ 
such that $\lambda (4/3)^{i-1} < |\D'| \leq \lambda (4/3)^{i}$. If 
$|\D'| < \lambda$, we define its level to be $0$ -- thus a base case input
$(S',\B')$ has level $0$. The following lemma bounds the quality of 
approximation of our non-deterministic algorithm.

\begin{lemma}
\label{lem:approxratio}
Assume that $\delta < 3/4 - 2/3$. There is an instantiation of the 
non-deterministic separator choices for which $\compcover(S',\B')$
computes a disjoint cover of size at most $(1 + \delta)^i |\D'|$,
where $i$ is the level of $(S',\B')$, and $\D'$ is an optimal disjoint
subset of $\B'$ that covers $S'$.
\end{lemma}

\begin{proof}
The proof is by induction on $i$. The base case is when $i = 0$, and here the statement follows from the base case of the algorithm. So assume that $i > 1$, and
that the statement holds for instances with level at most $i - 1$. 

Let $K' = |\D'|$. Suppose that the algorithm picks a separator $C_j \in \C$
that satisfies the guarantees of Lemma \ref{lem:sep2} when applied to
$\D'$. With this choice of $C_j$, let $S'_{j1}$,$\B'_{j1}$ $S'_{j1}$, and
$\B'_{j1}$ denote the same sets as in the algorithm.

We will show that there are sets $\D'_1 \subseteq \B'_{j1}$ and $\D'_2 \subseteq \B'_{j2}$ such that (a) $\D'_1$ (resp. $\D'_2$) is a pairwise disjoint
cover of $S'_{j1}$ (resp. $S'_{j2}$); (b) $|D'_1| \leq (2/3 + \delta) K'$,
and $|D'_2| \leq (2/3 + \delta) K'$; and (c) $|D'_1| + |D'_2| \leq 
(1 + \delta) K'$. 

Since $|D'_1| \leq (2/3 + \delta) K' \leq 3K'/4$, the level of
$(S'_{j1},\B'_{j1})$ is at most $i - 1$. By the inductive hypothesis,
$\compcover(S'_{j1},\B'_{j1})$ returns a solution of size at most 
$(1 + \delta)^{i-1}|D'_1|$. By the same reasoning, 
$\compcover(S'_{j2},\B'_{j2})$ returns a solution of size at most 
$(1 + \delta)^{i-1}|D'_2|$. It follows that the size of the
solution returned by  $\compcover(S',\B')$ is at most
\[ (1 + \delta)^{i-1}(|D'_1| + |D'_2|) \leq (1 + \delta)^i K'.\]

It remains to construct the sets $D'_1$ and $D'_2$. Let $\Dbad'$ denote
the set of those triangles in $\D'$ whose relative interiors are intersected
by $C_j$. Initialize a set $\Dnew'$. Take each triangle in $\Dbad'$, and retriangulate it so that the relative interior of each of the new triangles is not
intersected by $C_j$. In the retriangulation, the total number of triangles,
over all of $\Dbad'$, is proportional to $n(C_j,\D')$. For
each new triangle $\triangle$ of the retriangulation, add the set of at most
four pairwise disjoint basis triangles in $\B(\triangle)$ to $\Dnew'$ -- these four triangles
cover $S \cap \triangle$. We calculate that
$|\Dnew'| \leq 8 n(C_j, \D') \leq \delta K'$. Let $D'_1$ consist of those
triangles in $D' \setminus \Dbad'$ that are inside $C_j$ and those
triangles in $\Dnew'$ that are inside $C_j$. Thus we have $|D'_1| \leq
2 K'/3 + |\Dnew'| \leq (2/3 + \delta) K'$. Since $D' \setminus \Dbad' 
\cup \Dnew'$ covers $S'$, it follows that $D'_1$ covers $S'_{j1}$. Since
$\Dnew' \subseteq \B'$ (the closure property), it follows that $D'_1 \subseteq \B'_{j1}$. 

Similarly, letting  $D'_2$ consist of those
triangles in $D' \setminus \Dbad'$ that are outside $C_j$ and those
triangles in $\Dnew'$ that are outside $C_j$, we can establish similar 
properties for $\D'_2$.  Finally,
\[ |\D'_1| + |\D'_2| \leq |\D'| + |\Dnew'| \leq (1 + \delta) K'.\]
\end{proof} 

\paragraph{Deterministic Algorithm.}

The level of the input $(S, \B)$, where $S$ is the original set of points and
$\B$ the set of basis triangles, is clearly at most $n$, the size of $S$. It
follows that with $\compcover(S,\B)$, for suitable non-deterministic 
separator choices, returns a disjoint cover of size at most $(1 + \delta)^
{\lceil \log_{4/3}n \rceil}$ times the size of the optimal disjoint cover. Furthermore, the depth of the recursion with
such seperator choices is at most ${\lceil \log_{4/3}n \rceil}$.

To get a deterministic algorithm, we make the following natural changes.
If a call to $\compcover($ $S',\B')$ is at a recursion depth that is greater
than ${\lceil \log_{4/3}n \rceil}$ (with respect to the root corresponding to
$\compcover(S,\B)$), we return a special symbol $I$.  
In the $\compcover(S',\B')$ routine, when we are not in the base case, 
we try all possible separators $C_j \in \C$ instead of nondeterministically 
guessing one -- we return the smallest sized set $\left( \compcover(S'_{j1},\B'_{j1}) \cup \compcover(S'_{j2},\B'_{j2}) \right)$, over all $j$ for which neither of the the two recursive calls
returns $I$. If no such $j$ exists, $\compcover(S',\B')$ returns $I$.

With these changes, $\compcover(S,\B)$ is now a deterministic algorithm
that returns a disjoint cover of size at most $(1 + \delta)^
{\lceil \log_{4/3}n \rceil}$ times the size of the optimal disjoint cover. 
Its running time is 
\[ \left( n^{O(1/\delta^2)} \right)^{{\lceil \log_{4/3} n \rceil}} \cdot
   n^{\lambda} = n^{O\left( (\log n + \log 1/\delta)/\delta^3 \right)} .\]

Plugging $\delta = 1/{\lceil \log_{4/3}n \rceil}$, the approximation
guarantee for disjoint cover is $O(1)$ and the running time is
$n^{O(\log^4 n)}$. We can thus conclude with our main result for
surface approximation:

\begin{theorem}
There is an algorithm that, given inputs $\Sb$ and $\mu$ to the surface
approximation problem, runs in time $n^{O(\log^4 n)}$ and returns a solution
with complexity that is at most $O(1)$ times that of the optimal solution.
Here, $n$ is the number of points in $\Sb$.
\end{theorem}

\section{Discussion}
Consider the version of the convex decomposition problem where we are allowed
to add Steiner vertices. Can we obtain a QPTAS for this version of the problem?
Unlike the diagonal-based version we have considered, the separator for this version does not have to be made up of diagonals. In this sense, the Steiner version 
is simpler. The complication in the Steiner version is that we do not
know about the location of the Steiner points. Some way to bound their 
locations is needed to obtain, within a reasonable time bound, a suitable
separator family. In our surface approximation problem, we avoid confronting
this problem by losing a constant factor and reducting to the disjoint cover
problem.

In the convex decomposition problem, one often wants the convex pieces to
satisfy some additional criterion -- such as having an area that is at most
a specified quantity. For a diagonal based
decomposition, our separator construction goes through without modifications.
However, it is not clear if we can argue that
there is a near-optimal decomposition that respects the separator. This is
because the construction of the near-optimal decomposition needs to maintain
the additional criterion. It would be interesting to see if the argument
can be made to go through for certain criteria.

We hope that our work inspires progress along both these directions.

\bibliographystyle{plain}
\bibliography{SurfaceDecomp}

\end{document}